\documentclass[11pt,a4paper,oneside,reqno]{amsart}
\usepackage[utf8]{inputenc}

\usepackage[foot]{amsaddr}
\usepackage{amssymb,amsmath,amsfonts,amsthm,mathtools,color}
\usepackage{geometry}
\usepackage{graphicx}
\usepackage{mathrsfs}
\usepackage{bm}
\usepackage{enumerate}
\usepackage{url}
\usepackage{hyperref}

\newtheorem{theorem}{Theorem}[section]
\newtheorem{prop}{Proposition}[section]

\newtheorem{corollary}[prop]{Corollary}

\theoremstyle{definition}

\makeatletter
\@namedef{subjclassname@2020}{\textup{2020} Mathematics Subject Classification}
\makeatother

\begin{document}
\title{Ergodic optimal liquidations in DeFi}

\author{Jialun Cao$^1$}
\email{Galen.Cao@ed.ac.uk}

\author{David \v{S}i\v{s}ka$^{1,2}$}
\email{D.Siska@ed.ac.uk}

\address{$^1$School of Mathematics, University of Edinburgh, United Kingdom}
\address{$^2$Vega Protocol, \href{https://github.com/vegaprotocol/}{\texttt{https://github.com/vegaprotocol/}}}

\date{\today}

\begin{abstract}
We address the liquidation problem arising from the credit risk management in decentralised finance (DeFi) by formulating it as an ergodic optimal control problem.
In decentralised derivatives exchanges, liquidation is triggered whenever the parties fail to maintain sufficient collateral for their open positions. 
Consequently, effectively managing and liquidating disposal of positions accrued through liquidations is a critical concern for decentralised derivatives exchanges.
By simplifying the model (linear temporary and permanent price impacts, simplified cash balance dynamics), we derive the closed-form solutions for the optimal liquidation strategies, which balance immediate executions with the temporary and permanent price impacts, and the optimal long-term average reward.
Numerical simulations further highlight the effectiveness of the proposed optimal strategy and demonstrate that the simplified model closely approximates the original market environment. 
Finally, we provide the method for calibrating the parameters in the model from the available data.  

\end{abstract}

\keywords{Ergodic control, Optimal execution, Decentralised finance, Exchanges risk management}

\subjclass[2020]{Primary 93E20; Secondary 91G80}


\maketitle


\section{Introduction}

Credit risk, the possibility that a counterparty fails to fullfill their obligations in a contract, is a critical concern in both decentralised derivatives exchanges and lending protocols within the decentralised finance (DeFi) ecosystem. 
Without effective credit risk management, defaults can occur that threaten the solvency and stability of exchanges and protocols.

In decentralised derivatives exchanges,
traders must maintain sufficient collateral, a.k.a. margin, to mitigate credit risk in a trust-minimised environment. 
The minimum margin level is set so as to ensure that mark-to-market losses from leveraged or short positions can be covered with very high probability~\cite{siska2019margins}.
If there is a larger loss the network insurance pool is utilised and once that is depleted the full mark-to-market gains of parties on the winning side of a market move cannot be made (often referred to as loss socialisation).
Margin requirements are continuously monitored and if a party’s margin level falls below a specified threshold, for example the ``maintenance margin level''~\cite{vegaprotocol2020whitepaper} or ``maintenance margin percentage''~\cite{juliano2018dydx}, they are considered distressed and, depending on the decentralised exchange design, the position is transferred to the exchange itself and collateral in the margin account is confiscated 
or the exchange makes an immediate trade against the available volume to close (reduce the position)
or the position is open for liquidation by third parties.
Any third party can then bid to take over the position together with the collateral.
See~\cite[Section 6.1]{vegaprotocol2020whitepaper} and~\cite[Section 3.1]{juliano2018dydx} for the detailed mechanism designs.

For decentralised lending protocols, such as Aave~\cite{aave2020whitepaper}, Compound~\cite{compound2019whitepaper} and MakerDAO~\cite{makerdao_whitepaper}, credit risk arises when the value of the borrower's collateral drops below the loan value. 
To protect the lenders and maintain the solvency of the protocols, all loans are  over-collateralised at inception, i.e. the current mark-to-market value of the collateral value exceeds the loan amount.
However, due to the price fluctuations, the collateral may decrease in value and the debt amount increases with interest, leading to insolvency risk.
When the collateral value below the liquidation ratio, see~\cite[Liquidations]{aave2020whitepaper},
the collateral is open to liquidation by third parties.
The details vary by protocol. 
MakerDAO uses a Dutch auction to offer the collateral and debt for liquidation, while most lending protocols, like Aave and Compound, allow external liquidators to repay the debt in return for collateral at a discounted rate or with a liquidation bonus as incentives~\cite{qin2021empirical}. 

There is strong evidence, both theoretical and empirical that the design where third parties get bonus for taking over a distressed position are open to manipulation, termed ``adversarial liquidations''.
We refer to Cohen et al.~\cite{cohen2023paradox} that analyse the occurrence of adversarial liquidations, such as front-running or price manipulation, in the system and to Perez et al.~\cite{perez2021liquidations} and Qin et al.~\cite{qin2021empirical} for the
empirical analysis of liquidation risks in main decentralised lending protocols.  

To avoid adversarial liquidations the exchange itself has to take on the position and then either immediately dispose of it against available volume or hold onto it with the view of disposing of it gradually later. 
Immediate disposal is problematic for at least three reasons. 
First, it relies on liquidity which may not be present to handle large liquidations. 
Second, even modestly sized liquidation can lead to a large price move. 
This, depending on the mark-to-market methodology could make other parties' distressed.
Three, if the market is in an auction mode it may not execute trades immediately.
Hence a better design is to let the network dispose of accrued positions gradually.
This leads to the following natural question.

{\em How to design an optimal strategy which will minimise risks for the network and maximise the gains from disposing of the inventory acquired, taking into account that future liquidation event may significantly change the inventory? }

In order to answer this question:
\begin{enumerate}[i)]
\item We set up a stochastic model and an ergodic control problem for the real world problem described. 
\item We prove that the optimal strategy is to always dispose of a fraction of the inventory per unit of time. 
\item We give a precise expression for this fraction in terms of the model parameters and we also give analytic expression for the average earning per unit time.  
\item Finally, we show that all the model parameters can be easily estimated from available data and the only free parameter is risk-aversion.
\end{enumerate}

\subsection{Existing literature}
The ergodic optimal liquidation problem we will formulate is new. 
However, optimal liquidation (or acquisition) of a fixed position (amount) 
has received significant attention in the literature.
As a typical multi-stage decision problem, the optimal execution problem was initially formulated within the framework of stochastic control theory by Bertsimas and Lo \cite{bertimas1999optimal}, and Almgren and Chriss \cite{almgren2001optimal}.
This work was further extended into a continuous-time setting by 
Huberman and Stanzl~\cite{huberman2005optimal} and He and Mamaysky~\cite{he2005dynamic}. 
Since then, the framework has been extensively studied and enhanced to incorporate various realistic features, including temporary and permanent price impacts, stochastic price impact dynamics, net order flow in asset price dynamics, liquidation with different types of orders and so on, see~\cite{he2005dynamic, cartea2015algorithmic, cartea2016incorporating, gueant2016financial, barger2019optimal} and the references therein.

To address the optimal liquidation challenge in decentralised derivatives exchanges and liquidators, this paper adopts an ergodic optimal control framework, which was first studied by Arisawa and Lions in their seminal work~\cite{arisawa1998ergodic}. 
Additionally, we refer to~\cite{barles2001space, fujita2006asymptotic, ichihara2012large, ley2014large} for research on various types of Hamilton--Jacobi--Bellman (HJB) equations, either semi-linear or non-linear, in the ergodic setting, and to~\cite{gueant2020optimal, cao2024logarithmic} for further work on the continuous-time problems with discrete state spaces. 

In this paper, we primarily focus on the liquidation problem faced by decentralised derivatives exchanges; however, the model can also be adapted for use in lending protocols with some modifications: for example, rather than interpreting $r$ as the inverse of leverage, $r$ could represent the liquidation bonus percentage  in the cash process~\eqref{eq:cash_balance} introduced later.

\subsection{Ergodic formulation of the optimal liquidation problem}

To formulate this liquidation problem, we consider that the exchange (agent) seeks to maximise the long-term average reward from disposal of inventory while managing the risks associated with holding position over time.

The trading speed, controlled by the agent, is denoted by $(\nu_t)_{t \geq 0}$ taking values in $\mathbb R$, representing the volume of market orders executed per unit time. 
The profit or loss from holding a position is determined by $(S_t^\nu)_{t\geq 0}$ which is the mid-price (and mark-to-market) process driven by a standard Brownian motion $(W_t)_{t \geq 0}$.
It is subject to permanent price impact $g(\nu_t)$ due to trades our agent executes.
When the agent trades they trade at $(\hat S_t^\nu)_{t\geq 0}$ which is the execution price, defined as the midprice adjusted by subtracting a temporary price impact $f(\nu_t)$.
This temporary price impact reflects the bid/ask spread as well ``slippage'' or ``walking the book'' when executing larger orders.
At any time $t\geq 0$ the agent (exchange) holds inventory $(Q_t)_{t\geq 0}$ which is again an $\mathbb R$-valued process. 
We propose the following controlled dynamics
\begin{align}
dS^{\nu}_t & = - g(\nu_t) dt + \sigma dW_t, \quad S_0^{\nu} = S \, , \label{eq:midprice} \\
\hat S_t^{\nu} & = S_t^{\nu} -  f(\nu_t), \label{eq:execution_price} \\ 
d Q_t^{\nu} & = - \nu_t dt + \zeta^{+}_{N_t^{+}} dN_t^{+} - \zeta^{-}_{N_t^{-}} dN_t^{-}, \quad  Q_0^{\nu}=q \, , \label{eq:open_positions}  \\
d X_t^\nu & = \hat S_{t}^{\nu} \nu_t dt + r \zeta^{+}_{N_t^{+}} S_{t-}^{\nu}dN_t^{+}+ r \zeta^{-}_{N_t^{-}}S_{t-}^{\nu} dN_t^{-}, \quad X_0^\nu = x \, . \label{eq:cash_balance}
\end{align}
Here $(X^\nu_t)_{t\geq 0}$ is the cash balance of the agent.
On top of the effect of executed trades due to  $(\nu_t)_{t \geq 0}$ the cash balance and inventory are additionally impacted by liquidation events.
Theses processes are governed by $N_t^{\pm}$, two independent Poisson processes, also independent of $(W_t)_{t \geq 0}$, with intensities $\lambda^\pm$.
Additionally, the sizes of the long / short positions transferred from distressed parties to the exchanges when a liquidation event occurs are modelled by $(\zeta^+_k)_{k\in \mathbb N}$ and $(\zeta^-_k)_{k \in \mathbb N}$, respectively. 
We assume that $\zeta^\pm_k$ are independent, identically distributed random variables (also independent of $(N_t^\pm)_{t\geq 0}$ and $(W_t)_{t\geq 0}$  with $\mathbb E[\zeta^\pm_k] = \eta^\pm$ and $\mathbb E[|\zeta^\pm_k|^2] < + \infty$ for $k\in \mathbb N$. 
Furthermore, whenever a liquidation occurs, the distressed trader's collateral is transferred to the exchange's market insurance pool. This results in a cash inflow of $r \zeta^{\pm} S_t^{\nu}$ in the cash process, where $r$ denotes the inverse of the leverage ratio set by the exchange.

The exchange’s objective is to find the optimal strategy, in the class admissible strategies which will be made precise later, to maximise the average earnings per unit time subject to a quadratic inventory penalty scaled by $\phi>0$ expressing their risk aversion:
\begin{align}
J(x, S, q; \nu) & = \lim_{T \rightarrow + \infty}  \frac{1}{T} \mathbb{E}_{q, x, S} \bigg[ \int_0^T \mathrm{d} (X^{\nu}_t  + S_t^{\nu} Q_t^{\nu}) - \phi \int_0^T (Q^{\nu}_t)^2 \, \mathrm{d} t \bigg]\, , \label{eq:objective_function}  \\
\gamma & = \sup_{\nu} J(x, S, q; \nu) \, . \label{eq:gamma_definition}
\end{align}
Here $\mathbb{E}_{q, x, S} [\cdot]$ represents expectation conditional on $Q_0 = q, X_0 = x, S_0 = S$. 
The optimal earnings per unit time are denoted $\gamma$ and will be referred to as the ergodic constant. 

In order to obtain a closed-form solution, we make a number of simplifications which we will now discuss. 
\begin{enumerate}[i)]
\item \label{simplification1} We will assume that the parties don't update the initial margin balance and thus replace~\eqref{eq:cash_balance} by
\begin{equation} \label{simplified cash process}
d X_t^\nu = \hat S_t^{\nu} \nu_t dt + r \zeta^{+}S_0 dN_t^{+}+ r \zeta^{-}S_0 dN_t^{-}, \quad X_0^\nu = x, 
\end{equation}
where $S_0 \in \mathbb R^+$ is the initial mid price. 
\item \label{simplification2} We assume that liquidation intensities and sizes are symmetric for long and short positions i.e. that  $\lambda^+ = \lambda^- =: \lambda$ and $\eta^+ = \eta^- =: \eta$. 
\item \label{simplification3} We will assume linear price impact functions, specifically that temporary price impact function is $f(\nu) = k\nu$ some constant $k\in \mathbb R^+$ and that the permanent price impact is $g(\nu) = b\nu$ for some constant $b\in \mathbb R^+$. 
\end{enumerate}
Empirical and statistical evidence~\cite{cartea2015algorithmic} support the choice of linear temporary price impact.
However some studies~\cite{said2017market} argue that the temporary price impact conforms is better represented as a concave function, for example $\nu \mapsto k\sqrt{\nu}$.
We believe analytic tractability is more important which is why we choose linear temporary price impact. 
The empirical analysis in the work \cite{cartea2015algorithmic, said2017market} provides some justification for modelling the permanent price impact can also be described as a linear function of the trading rate. 

Based on the simplifications and assumptions, we derive the closed-form expression for $\gamma$, see Theorem~\ref{ergodic constant from discounted infinite}, as
\begin{equation} 
\label{form of ergodic constant}
\gamma = 2 r \lambda \eta S_0 - \lambda \eta^2 b  - 2 \lambda \eta^2 \sqrt{k\phi} \, , 
\end{equation}
which demonstrates that the ergodic constant $\gamma$ is independent of any initial conditions and is determined solely by the market parameters. 
The ergodic optimal Markov control is given by 
\begin{equation} \label{feedback ergodic control}
\nu^{*}(q) = \sqrt{\frac{\phi}{k}} q\, ,
\end{equation}
i.e. that the optimal strategy is indeed a fraction of the current position and the coefficient is determined by the parameters, see Corollary~\ref{ergodic control}.
Note that while the optimal average earnings per unit time depend on all the model parameters, the optimal strategy only needs the risk aversion $\phi>0$ and temporary price impact $k\in \mathbb R^+$. 
This means that the optimal strategy is robust to parameter misspecification; only the temporary price impact has to be estimated correctly.
In Section~\ref{appendix B} we discuss how all the parameters can be obtained from available data.

\subsection{Our contribution}
In this paper, we formulate and solve the ergodic optimal liquidation problem arising from the credit risk management in DeFi. 
Our framework provides simple formulae allowing derivatives exchanges optimise the profit and loss of their insurance pool through optimal liquidation of positions from distressed parties.
Since the strategy is designed to minimise price impact the beneficial side-effect is to minimise market moves caused by liquidations. 

Through numerical simulations, we demonstrate that the proposed optimal strategy significantly outperforms sub-optimal strategies.
Additionally, we perform sensitivity analysis $\gamma$ with respect to the key model parameters. 
Furthermore, we simulate the performance by using the optimal control under both simplified and original market environments, showing that the simplified model closely approximates the more complex one reality while maintaining the convenience of providing closed-form solutions for the strategy and ergodic constant. 

\section{Results}
The probability space for the liquidation problem is 
\begin{equation} \label{probability space}
(\Omega, \mathcal{F}, \mathbb{P}) = \big( 
\Omega^{W} \times \Omega^{N} \times \Omega^\zeta, 
\mathcal{F}^{W} \otimes \mathcal{F}^{N} \otimes \mathcal{F}^\zeta,  \mathbb{P}^{W} \otimes \mathbb{P}^{N} \otimes\mathbb{P}^\zeta
\big) \, , 
\end{equation}
where the space $(\Omega^{W}, \mathcal{F}^{W}, \mathbb{P}^{W})$ supports the Brownian motion $(W_t)_{t \geq 0}$ that drives the mid price process $(S_t^{\nu})_{t \geq 0}$. 
The space $(\Omega^{N}, \mathcal{F}^{N}, \mathbb{P}^N)$ supports the two independent Poisson processes $(N_t^{\pm})_{t \geq 0}$. 
The space $(\Omega^\zeta, \mathcal{F}^\zeta, \mathbb P^\zeta)$ supports the two iid sequences $(\zeta^+_k)_{k\in \mathbb N}$ and $(\zeta^-_k)_{k \in \mathbb N}$ with $\mathbb E[\zeta^\pm_k] = \eta^\pm$ and $\mathbb E[|\zeta^\pm_k|^2] < + \infty$ for $k\in \mathbb N$. 
The filtration is $\mathbb F := (\mathcal F_t)_{t\geq 0}$, where $\mathcal F_t = \sigma(W_r; r\leq t) \vee \sigma(N_r^\pm; r\leq t)\vee \sigma(\zeta_{N^+_r}^+,\zeta_{N^-_r}^-;r\leq t)$.
We will now define the class of admissible controls in two steps. 
First, for any $T>0$, let
\begin{equation*}
\begin{split}
\mathcal{A}[0, T] := \Big\{ \nu : [0,T]  & \times \Omega \to \mathbb R | \, \nu  \text{ is $\mathbb F$-progressively measurable, $\mathbb E \int_0^T |\nu_t|^2 \, \mathrm{d} t  < + \infty$, }   \Big\}\,,
\end{split}
\end{equation*}
The admissible controls for the ergodic problem are 
$\mathcal{A} := \bigcap_{T > 0} \mathcal{A}[0, T]$.

\subsection{Existence of ergodic constant and optimal Markov control}
Recall that for a given control the ergodic reward functional $J: \mathbb R^3\times\mathcal A \to \mathbb R$ is defined by~\eqref{eq:objective_function}. 
Our aim is give a characterisation of the optimal long-run average reward, also known as the ergodic constant, $\gamma$ defined by~\eqref{eq:gamma_definition} and to construct an optimal feedback (Markov) control $\nu^{*}  \in \mathcal{A}[0, \infty)$.

Let us define the running reward function $F: \mathbb R^2 \to \mathbb R$ as 
\begin{equation} \label{running reward function}
\begin{split}
F(q, \nu)  =  - k \nu^2 - b\nu q - \phi q^2 + 2 \lambda \eta r S_0 \,.
\end{split}
\end{equation}
Under the assumptions and simplifications discussed above, the ergodic optimal liquidation problem can be reduced to the following form, see Section~\ref{Dimension reduction under model simplifications}, 
\begin{equation} \label{reduced ergodic control problem}
J(q; \nu) = \lim_{T \rightarrow + \infty}  \frac{1}{T} \mathbb{E}_{q} \bigg[  
\int_0^T F(Q_t^\nu, \nu) \, \mathrm{d} t 
\bigg]\, , \quad \gamma = \sup \Big\{ J(q;\nu): \nu \in \mathcal{A} \Big\} \, .
\end{equation}

We will solve the ergodic problem by solving the finite time horizon problem first.
To that end we introduce a quadratic penalty $G(q): \mathbb R \to \mathbb R^+$ for any position $q \in \mathbb R$ held at the terminal time given, for $q\in \mathbb R$, by
\begin{equation} \label{terminal}
G(q) = - \alpha q^2 \, .    
\end{equation}
where $\alpha > 0$ is the terminal penalty parameter. 
The value function $v: [0,T] \times \mathbb R \to \mathbb R$ for the finite-time-horizon problem is 
\begin{equation} \label{finite time model}
v(t, q; T) = \sup_{\nu \in \mathcal{A}[t, T]} \mathbb E_{t, q} \Big[ \int_t^T F(Q_s^{\nu}, \nu) \,  \mathrm d s +  G(Q_T^{\nu}) \Big]  \, ,
\end{equation}
where $F$ is the running reward function given by~\eqref{running reward function}. 
We will also need discounted infinite-time-horizon optimal value function
\begin{equation} \label{discounted infinite model}
v_\beta(q) = \sup_{\nu \in \mathcal{A}} \mathbb E \Big[
\int_0^{+\infty} e^{-\beta t} F(Q_t^\nu, \nu_t) \mathrm{d} t
\Big] \, ,
\end{equation}
with the discount factor $\beta > 0$. 

The well-posedness of the optimal liquidation problems in either finite-time-horizon or discounted infinite-time-horizon setting are discussed in Sections~\ref{proof of prop: solution to the finite model} and~\ref{proof theorem: existence of discounted infinite}. 
To analyse the ergodic optimal liquidation problem~\eqref{reduced ergodic control problem}, following the general method used in works such as~\cite{arisawa1998ergodic, cao2024logarithmic}, we let $\beta \rightarrow 0$ in the discounted infinite-time-horizon control problem~\eqref{discounted infinite model} and $T \rightarrow + \infty$ in the finite-time-horizon control problem~\eqref{finite time model}. We then demonstrate that the value $\beta v_{\beta}$ (and respectively, $\frac{1}{T} v(\cdot; T)$) converges to a constant independent of the  initial state, which is indeed the ergodic constant, as stated in the following theorem. The proof is provided in Section~\ref{proof theorem: ergodic constant from discounted infinite}.
\begin{theorem} \label{ergodic constant from discounted infinite}
There exists a constant $\hat \gamma \in \mathbb R$ such that, for all $q \in \mathbb R$, 
\begin{equation*}
\hat \gamma = \lim_{\beta \rightarrow 0} \beta v_{\beta}(q) = \lim_{T \rightarrow + \infty} \frac{1}{T} v(0,q; T) \,, 
\end{equation*}
where $v_\beta$ is given by~\eqref{discounted infinite model} and $v$ is given by~\eqref{finite time model}. 
Moreover, we have 
\begin{equation*}
\gamma = \hat \gamma = 2 r \lambda \eta S_0 - \lambda \eta^2 b  - 2 \lambda \eta^2 \sqrt{k\phi} \, , 
\end{equation*}
where $\gamma$ is the ergodic constant given~\eqref{reduced ergodic control problem}.
\end{theorem}

The following corollary provides the expression for the optimal ergodic control, with the proof given in Section~\ref{proof prop optimal ergodic control}.
\begin{corollary} 
\label{ergodic control}
The optimal Markov control for the ergodic control problem~\eqref{reduced ergodic control problem} is 
\begin{equation*}
\nu^{*}(q) = \sqrt{\frac{\phi}{k}} q\, .
\end{equation*}
\end{corollary}

\subsection{Numerical results}
In this section, we simulate the ergodic liquidation problem, compare the performance between the optimal ergodic control and not well-calibrated controls, and explore the impact of the model parameters on the ergodic constant $\gamma$ which indicates the average earnings / loss of the insurance pool per unit time. 

\begin{figure}
\centering
\includegraphics[height=150pt]{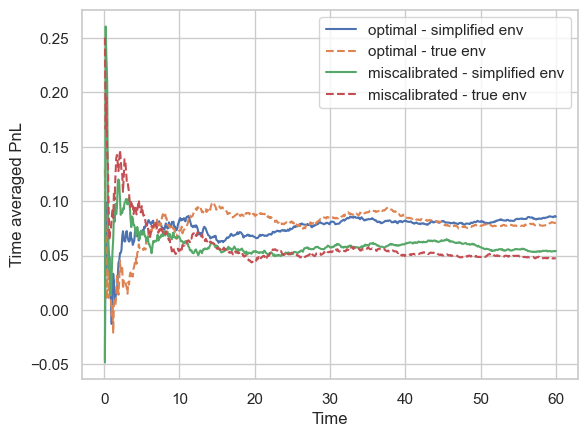}
\includegraphics[height=150pt]{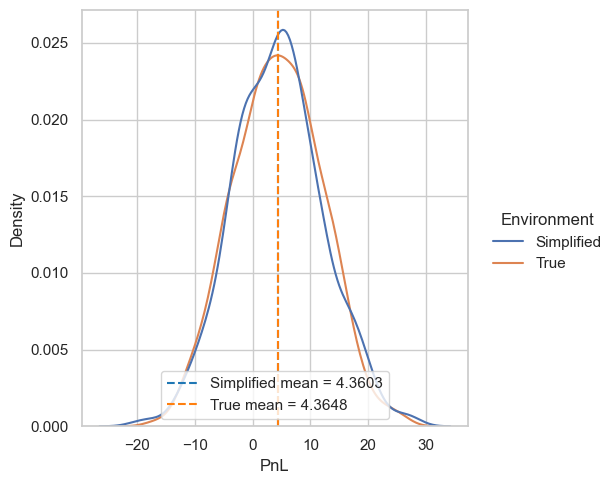}
\caption{Left: Comparison of the time-averaged PnL for optimal and mis-calibrated controls. 
Right: PnL distributions due to the optimal ergodic control~\eqref{feedback ergodic control} derived using simplified dynamics~\eqref{simplified cash process} under the true dynamics~\eqref{eq:cash_balance} - in yellow - and  simplified dynamics~\eqref{simplified cash process} - in blue.}
\label{pnl}
\end{figure}

We first run the simulations by employing the ergodic optimal control~\eqref{feedback ergodic control} and a not well-calibrated control under the simplified environment and original environment, respectively. 
The mis-calibrated control is simple: in each time period the exchange trades volume equal half of the current open positions in the direction which will reduce its position.
The simplified (or original) environment means that whether the cash process $(X_t)_{t \geq 0}$ is updated by the dynamics~\eqref{simplified cash process} (or~\eqref{eq:cash_balance}). 
The parameters in the simulation are set as follows: 
$\lambda^{\pm}=0.05/s$, $\zeta^{\pm} \sim \mathcal{N}(10, 0.5)$, $\sigma=0.5 s^{-1/2}\$$, $b = \$1 \times 10^{-5}$, $k = \$1 \times 10^{-3}$, $\phi = \$ 1 \times 10^{-4}$, $S_0 = \$10$ and $Q_0 = 0$.
Figure~\ref{pnl} (left penal) plots the achieved time-averaged PnL of each control in the simplified and original environment. Clearly, the ergodic optimal control gains a higher long-term average reward than that of using the heuristic (sub-optimal) control.

The right panel of Figure~\ref{pnl} plots the distributions of the achieved PnL after a long period by employing the ergodic control in
both simplified and original environments. The dashed vertical lines present the mean values of the PnL for each, indicating a negligible difference.
The resemblance between the distributions suggests that the simplified dynamics~\eqref{simplified cash process} provides a good approximation of the dynamics~\eqref{eq:cash_balance}.

\begin{figure}
\centering
\includegraphics[height=3.8cm]{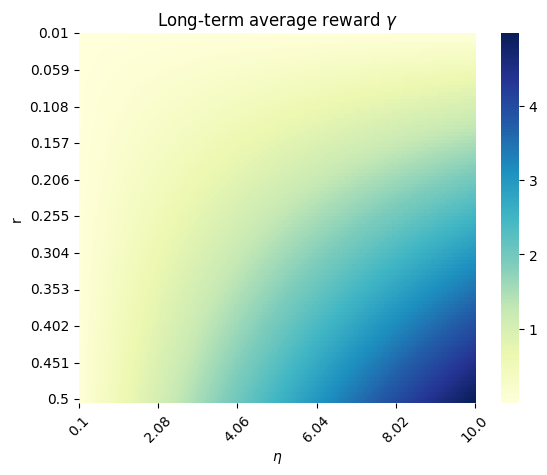}\quad
\includegraphics[height=3.8cm]{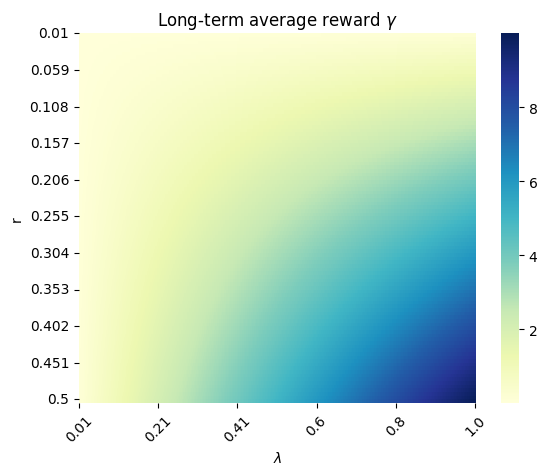}

\medskip
\includegraphics[height=3.8cm]{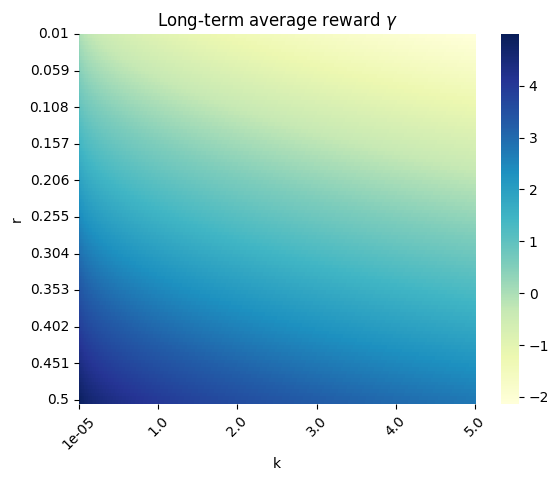}\quad
\includegraphics[height=3.8cm]{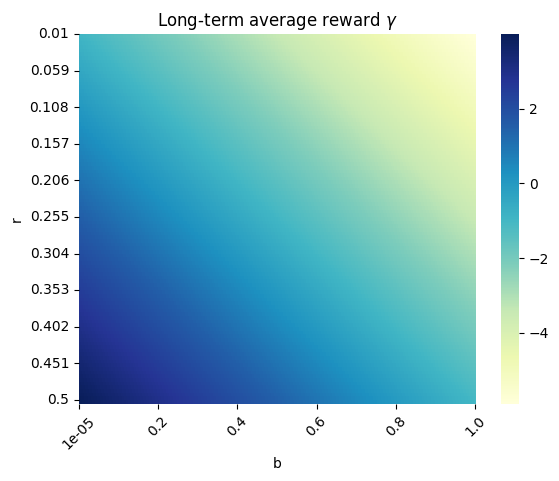}

\caption{
{
The long-term average reward $\gamma$ as a function of inverse of leverage $r$ and: average size of distressed positions $\eta$ (top left), intensity of liquidation occurrence $\lambda$ (top right), temporary price impact $k$ (bottom left), permanent price impact $b$ (bottom right) all from~\eqref{form of ergodic constant}.}
}
\label{model parameter}
\end{figure}

The long-term average reward $\gamma$ derived from the simplified model is influenced by several parameters, including the inverse of the leverage ratio $r$, the average size of distressed positions $\eta$, the intensity of liquidation occurrence $\lambda$, temporary price impact $k$ and permanent price impact $b$. 
Figure~\ref{model parameter} illustrates the behaviour of $\gamma$ as a function of the parameters.
From the figures, we can see that increasing $\eta$ and $\lambda$ tends to improve $\gamma$, indicating that both a larger average size of distressed positions and a higher frequency of liquidations lead to higher long-term rewards for the exchange as long as $r$ is sufficiently large. 
However, both temporary and permanent price impacts $k$ and $b$ negatively affect $\gamma$, which means that higher adverse market impacts reduce overall profitability when liquidating. 
Unsurprisingly, across all figures, higher leverage, i.e. lower $r$, diminishes the optimal average earnings, indicating that exchanges should manage leverage carefully to ensure the insurance pool has a healthy balance.

While~\eqref{form of ergodic constant} tells us that $\gamma$ derived using the simplified dynamics~\eqref{simplified cash process}  does not depend on the volatility $\sigma$ a question remains whether this is only a result of the simplification.
The left panel of Figure~\ref{fig:gamma_sigma} is a result of Monte-Carlo simulations using the control~\eqref{feedback ergodic control} under the original dynamics~\eqref{eq:cash_balance} to obtain the long run average earnings versus various values of $\sigma$ and $r$. 
The plot reveals that the volatility parameter $\sigma$ has no clear impact on the long-term average reward $\gamma$ even under the original dynamics~\eqref{eq:cash_balance}. 
This suggests that the closed-form expression of $\gamma$ derived under the simplified model, which indicates that volatility $\sigma$ does not directly affect $\gamma$, is a good approximation for the behaviour of the long-term average rewards in the original environment.

\begin{figure}
\centering
\includegraphics[height=3.8cm]{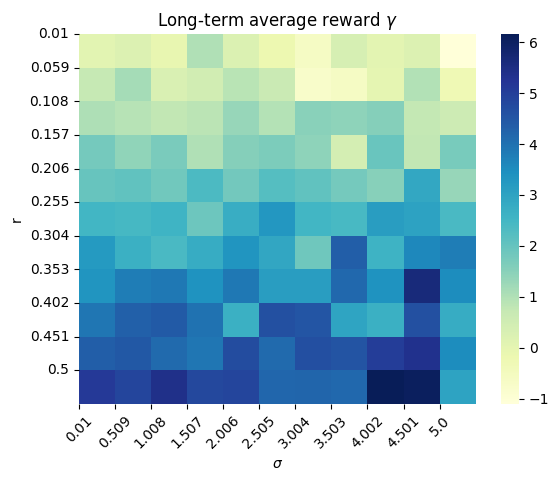}
\includegraphics[height=3.8cm]{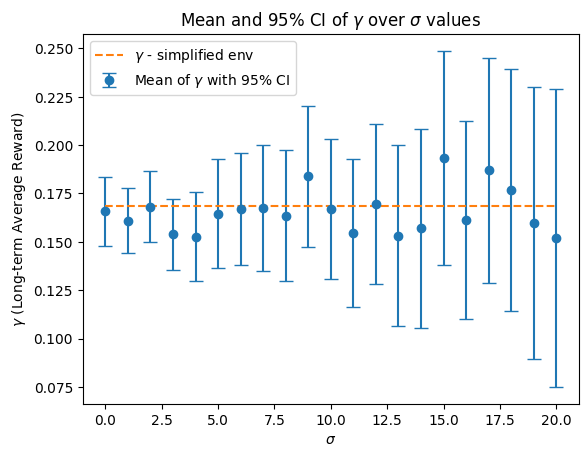}
\includegraphics[height=3.8cm]{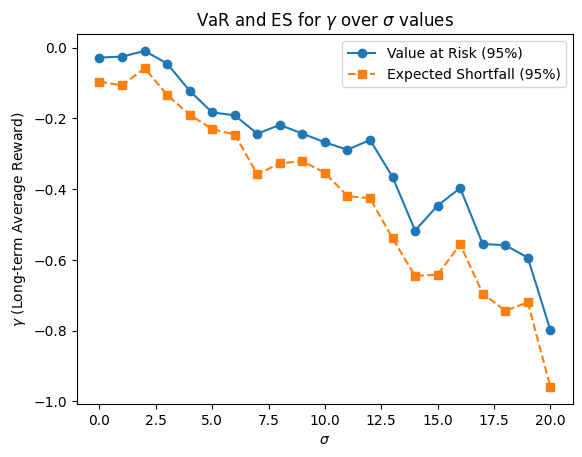}
\caption{
The long-term average reward $\gamma$ as a function of inverse of leverage $r$ and volatility $\sigma$ (left) by using Monte-Carlo simulations under the original dynamics~\eqref{eq:cash_balance} with various volatility $\sigma$ and inverse of leverage $r$. The middle and right panels plot the
mean and $95\%$ confidence intervals (middle) and Value at Risk and Expected Shortfall (right) of the long-term average reward $\gamma$ across various  volatility $\sigma$ under the original dynamics~\eqref{eq:cash_balance}.
While we do know that $\gamma$ given by~\eqref{form of ergodic constant} from the simplified dynamics~\eqref{simplified cash process} does not depend on $\sigma$ a question remains to what extent this is a consequence of the simplification. 
From the plots is seems that even under~\eqref{eq:cash_balance} the expectation of long-term average reward $\gamma$ shows negligible dependence on $\sigma$, whereas the risk of loss increases with rising volatility.}
\label{fig:gamma_sigma}
\end{figure}

Figure~\ref{fig:gamma_sigma} (middle) displays the mean of the long-term average reward $\gamma$ over different values of $\sigma$, with 95\% confidence intervals under the original dynamics~\eqref{eq:cash_balance} . 
Overall, the mean of $\gamma$ remains relatively stable as $\sigma$ increases, fluctuating around the dotted line, which is the theoretical value of $\gamma$ from~\eqref{form of ergodic constant} under the simplified model. 
We also see that it is just about within the 95\% confidence bounds arising from using a finite number of Monte-Carlo trajectories in simulation.
The right panel plots Value at Risk (VaR) and Expected Shortfall (ES) for $\gamma$ over different $\sigma$.
We can see that both VaR and ES decline as $\sigma$ increases.
In other words: higher volatility increases risk but doesn't impact the average long-term reward.

\section{Parameter calibration}
\label{appendix B}
In this section, we provide an example of possible approach to parameter calibration. 
There are nearly endless number of choices to be made in the methodology and we do not claim that the approach below is the best. 
The aim is to show that the calibration is possible using classical statistical tools and argue that the data required is available from decentralised derivative exchanges.

Let $[0, T]$ be a fixed time interval, e.g. 1 hour.
We collect the times at which liquidations of distressed traders occur over $[0,T]$, denoted by $(\tau_i)_{i=1}^N \in [0,T]$, where $N$ represents the total number of liquidation events observed.
Additionally, we record the absolute value of the size of the distressed positions at each liquidation, represented by $(\zeta_i)_{i = 1}^N$. By using maximum likelihood estimators (MLE), see~\cite[Section 7.2]{casella2024statistical}, we estimate the intensity of liquidation occurrence $\lambda$ and the average size of distressed positions $\eta$ as 
$\lambda = \frac{N}{\tau_N}$ and $\eta = \frac{1}{N} \sum_{i=1}^N \zeta_i$.

To calibrate the temporary price impact parameter $k$, we divide the time interval $[0, T]$ into multiple subintervals $\sum_i [t_i, t_{i+1}]$, for example, segments of 5 minutes each. 
Within each subinterval, we take a snapshot of the buy (sell) side of the limit order book at time $t_i$.
Let $(P_l^{(i)}, V_l^{(i)})_{l=1}^N$ denote the price levels and corresponding volumes of limit orders in the limit order book, which the superscript represents the time index. 
For the purpose of regression, we define a set of fixed trade sizes $(Q_j)_{j=1}^M$, such as $Q_1 = 1, Q_2 = 2, \dots, Q_M = 100$.
By walking through the limit order book, we determine the execution price per unit size, $\hat P_j^{(i)}$, for each order with volume $Q_j$ as
\begin{equation*}
\hat P_j^{(i)} = \sum_l \frac{P_{l}^{(i)} V_{l}^{(i)}}{Q_j} \, , \quad \text{with } \sum_l V_{l}^{(i)} = Q_j \, .
\end{equation*}
We then perform a linear regression of $|\hat P_j^{(i)} - S^{(i)}|$ against $Q_j$, where $S^{(i)}$ is the midprice at time $t_i$. The slope of this linear regression, denoted $k^{(i)}$, represents the calibrated temporary price impact parameter for the interval $[t_i, t_{i+1}]$. The overall parameter $k$ is taken as the mean value of $k^{(i)}$ across all subintervals, i.e. $k = \frac{1}{M} \sum_{i=1}^M k^{(i)}$.

To estimate the permanent price impact parameter $b$ we again divide the time interval $[0, T]$ into multiple subintervals $\sum_i [t_i, t_{i+1}]$. Let $\mu_i$ be the net order-flow in the subinterval $[t_i, t_{i+1}]$ and $\Delta S_i = S_{t_{i+1}} - S_{t_i}$ be the change in the midprice. 
As we assume the permanent price impact is linear in the model, hence we can perform the linear regression 
\begin{equation*}
\Delta S_i = b \mu_i + \varepsilon_i \, ,
\end{equation*}
where the error terms $\varepsilon_i$ are assumed to be iid with a normal distribution. The slope from this regression provides the calibration for $b$.

\section{Conclusion}

We have introduced an ergodic control problem describing how a decentralised derivatives exchange should manage disposal of positions accrued through liquidations (closeouts) of parties that fail to meet their margin requirement. 
We simplify the model (linear temporary and permanent price impacts, simplified cash balance dynamics) so that we can solve the ergodic stochastic control problem to obtain closed form solutions for the Markov optimal control and optimal long term rewards. 
The simplifications are in our opinion justified by existing literature on price impact and the numerical experiments performed.  
The optimal strategy can be easily implemented and moreover the expression for the long term average reward gives, after calibration of parameters to available data, an indication of how the insurance pool balance will evolve given the required margin fraction $r$ (inverse leverage).

The model is of course a crude simplification of reality and can be criticised from a number of angles, in particular we don't expect the mid-price to follow a Gaussian distribution, and we shouldn't expect price impacts to be linear. 
Nevertheless, we believe that the benefit of having obtained simple formulae that have clear interpretations is valuable.

\section{Proofs}

\subsection{Dimension reduction under model simplifications}
\label{Dimension reduction under model simplifications}
We will show that under the model simplifications~\ref{simplification1}), \ref{simplification2}) and~\ref{simplification3})
the control problem~\eqref{eq:objective_function} can be simplified to~\eqref{reduced ergodic control problem}.
First observe that
\begin{equation}
\begin{split}
\mathrm{d}(X_t^\nu + S_t^\nu Q_t^\nu) & = \mathrm{d} X_t^\nu  + Q_t^\nu \mathrm{d}S_t^\nu + S_t^\nu \mathrm{d}Q_t^\nu + \mathrm{d} S_t^\nu \mathrm{d} Q_t^\nu \\ 
& = (S_t^{\nu} - k \nu_t)\nu_t \mathrm{d}t + r \eta S_0 \mathrm{d}N_t^{+}+ r \eta S_0 \mathrm{d} N_t^{-} \\ & \qquad + Q_t^\nu (-b\nu_t \mathrm{d}t + \sigma \mathrm{d}W_t)  + S_t^\nu (- \nu_t \mathrm{d}t + \eta \mathrm{d} N_t^{+} - \eta \mathrm{d} N_t^{-})  \\ 
& = (- k \nu_t^2 - b\nu_t Q_t^\nu + 2 r S_0 \lambda \eta) \mathrm{d}t  + (\eta r S_0 + \eta S_t^\nu ) \mathrm{d} \tilde N_t^+  \\ 
& \qquad + (\eta r S_0 - \eta  S_t^\nu ) \mathrm{d} \tilde N_t^- + \sigma Q_t^\nu \mathrm{d} W_t \, , 
\end{split}
\end{equation}
where $\tilde N_t^{\pm}$ are independent compensated Poisson processes. 
If we can show that the expectation of all the stochastic integral terms is $0$ then  $J$ given by~\eqref{eq:objective_function} can be simplified as follows:
\begin{equation*}
    \begin{split}
        J(x, S, q; \nu) & = \lim_{T \rightarrow + \infty}  \frac{1}{T} \mathbb{E}_{q, x, S} \bigg[ \int_0^T \mathrm{d} (X^{\nu}_t  + S_t^{\nu} Q_t^{\nu}) - \phi \int_0^T (Q^{\nu}_t)^2 \, \mathrm{d} t \bigg] \\ 
        & = \lim_{T \rightarrow + \infty}  \frac{1}{T} \mathbb{E}_{q} \bigg[ \int_0^T \big(- k \nu_t^2 - b\nu_t Q_t^\nu + 2 r S_0 \lambda \eta  - \phi (Q^{\nu}_t)^2 \big) \, \mathrm{d} t \bigg] \\
        & = \lim_{T \rightarrow + \infty}  \frac{1}{T} \mathbb{E}_{q} \bigg[ \int_0^T F(Q_t^\nu, \nu_t) \big) \, \mathrm{d} t \bigg] =: J(q; \nu) \, , 
    \end{split}
\end{equation*}
where the running reward function $F: \mathbb R^2 \to \mathbb R$ is 
\begin{equation*} 
F(q, \nu) =  - k \nu^2 - b\nu q - \phi q^2 + 2 \lambda \eta r S_0 \,. 
\end{equation*}
Thus to argue that the control problem~\eqref{eq:objective_function} can be simplified to~\eqref{reduced ergodic control problem}
it just remains to show that the expectations of the stochastic integrals is indeed $0$.

As $S_t^\nu$ with dynamics~\eqref{eq:midprice} is $\mathcal{F}_t$-adapted and continuous, now we need to show that $\eta r S_0 \pm \eta S_t^\nu$ are square integrable so that the stochastic integrals with respect to the compensated Poisson processes are 0, see~\cite[Theorem 11.4.5]{shreve2004stochastic}. 
To see this, note that
\begin{equation*}
    \begin{split}
        \mathbb E \int_0^T (\eta r S_0 \pm \eta S_t^\nu)^2 \mathrm{d}t  \leq C T \sup_{t \in [0,T]} \mathbb E \big[(S_t^\nu)^2 \big] + C T \sup_{t \in [0,T]} \mathbb E \big[S_t^\nu \big] + CT \, ,
    \end{split}
\end{equation*}
where $C$ are constants independent of $T$ that are different for each term. 
Hence to show that $\mathbb E \int_0^T (\eta r S_0 \pm \eta S_t^\nu)^2 \mathrm{d}t < + \infty$, it is enough to show that $\sup_{t \in [0,T]} \mathbb E \big[(S_t^\nu)^2 \big] < + \infty$ due to the fact that $\mathbb E \big[S_t^\nu \big] \leq \sqrt{\mathbb E \big[(S_t^\nu)^2 \big]}$. 
By the dynamics~\eqref{eq:midprice}, we have $S_t^\nu = S_0  - b \int_0^t \nu_s \mathrm{d} s  + \sigma W_t$.
Hence 
\begin{equation*}
    (S_t^\nu)^2 = S_0^2 + \Big(\int_0^t \nu_s \mathrm{d}s\Big)^2 + \sigma^2 (W_t)^2 + 2 S_0 \sigma W_t - 2 \sigma \Big(\int_0^t \nu_s \mathrm{d} s \Big) W_t - 2 S_0 \int_0^t \nu_s \mathrm{d} s \, . 
\end{equation*}
Therefore, for any  $\nu \in \mathcal{A}$,
\begin{equation*}
    \begin{split}
        \sup_{t \in [0,T]} \mathbb E\big[(S_t^\nu)^2\big] & \leq S_0^2 + T \mathbb E\int_0^T |\nu_t|^2 \mathrm{d} t  + \sigma^2 T + 2 \sigma T \sqrt{\mathbb E\int_0^T |\nu_t|^2 \mathrm{d} t } + 2S_0 \sqrt{T \mathbb E\int_0^T |\nu_t|^2 \mathrm{d} t } \\
        & < + \infty \,, 
    \end{split}
\end{equation*}
hence we can conclude that $\mathbb E[\int_0^T (\eta r S_0 \pm \eta  S_t^\nu ) \mathrm{d} \tilde N_t^\pm ] = 0$. 

Moreover, by the dynamics~\eqref{eq:open_positions}, we have 
\begin{equation*}
    Q_t^\nu = Q_0 - \int_0^t \nu_s \mathrm{d} s + \sum_{k=1}^{N_t^+} \zeta_k^+ + \sum_{k=1}^{N_t^-} \zeta_k^- \, . 
\end{equation*}
Therefore, 
\begin{equation*}
\begin{split}
    \mathbb E[|Q_t^\nu|^2] & \leq 4 Q_0^2  + 4 t \mathbb E[\int_0^t |\nu_s|^2 \mathrm{d} s] + 4  \mathbb E\Big[ \big(\sum_{k=1}^{N_t^+} \zeta_k^+ \big)^2 \Big] + 4  \mathbb E\Big[ \big(\sum_{k=1}^{N_t^-} \zeta_k^- \big)^2 \Big]\\ 
     & \leq 4 Q_0^2  + 4 t \mathbb E[\int_0^t |\nu_s|^2 \mathrm{d} s] + 4 \sum_{- , +} \bigg(\sum_{n \geq 0} \mathbb E\Big[ \sum_{k, l}^{n} \zeta_k^\pm \zeta_l^\pm \Big| N_t^\pm = n \Big] \mathbb P(N_t^\pm = n)\bigg) \, ,
\end{split}
\end{equation*}
where $\sum_{-,+}$ is used to indicate the summation over both the $+$ and $-$ processes. 
Recall that $\nu \in \mathcal{A}$ is square integrable. 
Additionally, $\zeta^\pm_k$ are i.i.d with $\mathbb E[\zeta^\pm_k] = \eta^\pm$ and finite second moment and independent of the Poisson processes $N_t^\pm$. 
Hence we have 
\begin{equation*}
    \begin{split}
        \sup_{t \in [0,T]} \mathbb E[|Q_t^\nu|^2] & \leq 4 Q_0^2  + 4 T \mathbb E[\int_0^T |\nu_s|^2 \mathrm{d} s] \\
        & \qquad \qquad + 4 \sum_{- , +} \Big(\mathbb E \big[ N_t^\pm \big] \mathbb E \big[ |\zeta_1^\pm|^2 \big] + \Big( \mathbb E \big[ |N_t^\pm|^2  - N_t^\pm \big]  \Big) \mathbb E \big[\zeta_1^+ \zeta_2^+ \big] \Big) \\
        & \leq 4 Q_0^2  + 4 T \mathbb E[\int_0^T |\nu_s|^2 \mathrm{d} s] + 4  \sum_{- , +} \Big(\lambda^\pm T \mathbb E \big[ |\zeta_1^\pm|^2 \big] + (\lambda^\pm)^2 T^2 (\eta^\pm)^2 \Big) \\ 
        & < + \infty \, .
    \end{split}
\end{equation*}
Therefore the stochastic process $Q_t^\nu$ with dynamics~\eqref{eq:open_positions} is $\mathcal{F}_t$-adapted and square integrable for any $\nu \in \mathcal{A}$, hence $\mathbb E[\int_0^T \sigma Q_t^\nu \mathrm{d} W_t] = 0$.

\subsection{Finite-time-horizon control problem~\eqref{finite time model}} 
\label{proof of prop: solution to the finite model}
\begin{prop} \label{solution to the finite model}
Let $\varphi = \sqrt{\frac{\phi}{k}}$ and $\xi = \frac{\alpha - \frac{1}{2}b + k \varphi}{\alpha - \frac{1}{2}b - k \varphi}$. Let $h_0, h_2 \in C^1( [0, T], \mathbb R)$ be 
\begin{equation} \label{solution h_2 h_0 finite}
\begin{split}
h_2(t) & = k \varphi \frac{1 + \xi e^{2 \varphi (T-t)}}{1 - \xi e^{2 \varphi (T-t)}} - \frac{1}{2}b \, , \\
h_0(t) & = 2 \lambda \eta^2 k \ln  \frac{\xi - 1}{\xi e^{\varphi (T-t)} - e^{- \varphi (T-t)}} + (\lambda \eta^2 b - 2 r \lambda \eta S_0) (t-T) \, .
\end{split}
\end{equation}
Then the function $u: [0, T] \times \mathbb R \to \mathbb R$ given by 
\begin{equation} \label{solution}
u(t,q) = h_0(t) + h_2(t) q^2
\end{equation}
solves the HJB equation associated with the control problem~\eqref{finite time model} which is given by 
\begin{equation} \label{finite hjb}
0 = \partial_t u + \sup_{\nu \in \mathcal{A}[t, T]} \Big\{ - \nu \partial_q u + F(\nu, q) + \lambda \mathbb E\big[u(t, q + \zeta^+) - u(t,q) \big]+ \lambda \mathbb E\big[u(t, q - \zeta^-) - u(t,q) \big]  \Big\} \, ,
\end{equation}
with the terminal condition 
\begin{equation*}
u(T, q) = - \alpha q^2 \, .
\end{equation*}
\end{prop}

\begin{proof}
As the value function is given by~\eqref{finite time model}, by using the Dynamic Programming Principle, the associated HJB equation is
\begin{equation*}
\begin{split}
0 = & \partial_t u  + \sup_{\nu} \big\{ -\nu \partial_q u + F(\nu, q) \big\} +  \lambda \mathbb E\big[u(t, q + \zeta^+) - u(t,q) \big]+ \lambda \mathbb E\big[u(t, q - \zeta^-) - u(t,q) \big]   \big\} \, ,
\end{split}
\end{equation*}
where $F$ is given by~\eqref{running reward function}. By letting $\partial_\nu \big(  - \nu \partial_q u + F(\nu, q) \big) = 0$ and checking that $\partial_{\nu \nu} \big(  - \nu \partial_q v + F(\nu, q) \big) < 0$, we know that the supremum is attained at  
\begin{equation} \label{sup of control}
\nu^* =  \frac{-(bq + \partial_q u)}{2k} \, .
\end{equation}
Therefore, the HJB equation satisfies
\begin{equation} \label{trans HJB}
0 = \partial_t u  + \frac{1}{4k} (bq + \partial_q u)^2 - \phi q^2 + 2 r \lambda \eta S_0  +  \lambda \mathbb E\big[u(t, q + \zeta^+) - u(t,q) \big]+ \lambda \mathbb E\big[u(t, q - \zeta^-) - u(t,q) \big]   \, .
\end{equation}

To solve~\eqref{trans HJB}, we make the ansatz $u(t,q) = h_0(t) + h_1(t) q + h_2(t) q^2$ with the terminal conditions $h_0(T)= 0, h_1(T) = 0$ and $h_2(T)= -\alpha$. By plugging the ansatz into~\eqref{trans HJB}, we have the following coupled ODEs 
\begin{equation}
\begin{split}
\partial_t h_0 + \frac{1}{4k} h_1^2+ 2 \lambda \eta^2 h_2 + 2 r \lambda \eta S_0  & = 0 \, , \\
\partial_t h_1 + \frac{b + 2h_2}{2k} h_1 & = 0 \, ,\\
\partial_t h_2  + \frac{(b + 2h_2)^2}{4k} - \phi & = 0 \, .
\end{split}
\end{equation}

We start with solving $h_2(t)$. Let $h'_2(t) = h_2(t) + \frac{1}{2} b$ with the terminal condition $h'_2(T) = h_2(T) + \frac{1}{2}b = - \alpha + \frac{1}{2}b$. Therefore, 
\begin{equation*}
k \partial_t h'_2 - k \phi + (h'_2)^2 = 0 \, . 
\end{equation*}
We transform it into the following form since $\phi > 0$
\begin{equation*}
\begin{split}
\frac{\partial_t h'_2}{\sqrt{k\phi} - h'_2} + \frac{\partial_t h'_2}{\sqrt{k\phi}+h'_2} & = 2 \sqrt{\frac{\phi}{k}} \, . 
\end{split}
\end{equation*}
By integrating both sides over $[t,T]$ and using the terminal condition $h'_2(T) = \frac{1}{2}b - \alpha$, where we assume that $\alpha \gg b, k$, see~\cite[Chapter 6.5]{cartea2015algorithmic} for more details, then we obtain 
\begin{equation*}
h'_2(t) = k \varphi \frac{1 + \xi e^{2 \varphi (T-t)}}{1 - \xi e^{2 \varphi (T-t)}} \, ,
\end{equation*}
where $\varphi = \sqrt{\frac{\phi}{k}}$ and $\xi = \frac{\alpha - \frac{1}{2}b + k \varphi}{\alpha - \frac{1}{2}b - k \varphi} > 1$. Therefore, 
\begin{equation*}
h_2(t) = k \varphi \frac{1 + \xi e^{2 \varphi (T-t)}}{1 - \xi e^{2 \varphi (T-t)}} - \frac{1}{2}b \, . 
\end{equation*}

We then solve $h_1(t)$. Let us define $m(t) = \frac{b+2h_2(t)}{2k} = \frac{h'_2(t)}{k}$, then $h_1(t)$ satisfies 
$$\partial_t h_1 + m(t) h_1 = 0 \, , $$ with the terminal condition $h_1(T) = 0$. To solve the ODE, we have
\begin{equation*}
\begin{split}
\ln |h_1(t) | & = \int \frac{1}{h_1} \mathrm{d} h_1 = - \int m(t) \mathrm{d} t \\
& = - \frac{1}{k}\int h'_2(s) \mathrm{d} s \\ 
& = -  \int \varphi \frac{1 + \xi e^{2 \varphi (T-s)}}{1 - \xi e^{2 \varphi (T-s)}} \mathrm{d} s \, ,
\end{split}
\end{equation*}
By using the change of variables $x = \sqrt{\xi} e^{\varphi (T-s)}$ in the integration, we have 
\begin{equation*}
h_1(t) = A \frac{\xi e^{\varphi T} - e^{-\varphi T}}{\xi e^{\varphi(T-t)} - e^{-\varphi(T-t)}} \, , 
\end{equation*}
with $A \in \mathbb R$ that represents the general solution for the ODE. As the terminal condition is $h_1(T)=0$, therefore $A = 0$, hence $h_1(t) = 0$. 

We finally solve $h_0(t)$. Let $n(t) = 2 \lambda \eta^2 h_2(t)+ 2 r \lambda \eta S_0$, then $h_0(t)$ satisfies 
$$\partial_t h_0 + n(t)  = 0 \, , $$ with the terminal condition $h_0(T)= 0$. Therefore, 
\begin{equation*}
\begin{split}
h_0(T) - h_0(t) & = - \int_t^T n(s) \mathrm{d} s = - 2 \lambda \eta^2 \int_t^T   h'_2(s) \mathrm{d} s + \int_t^T (\lambda \eta^2 b - 2 r \lambda \eta S_0) \mathrm{d} s \\
& = - 2 \lambda \eta^2 \int_t^T    k \varphi  \frac{1 + \xi   e^{2 \varphi (T-s)}}{1 - \xi   e^{2 \varphi (T-s)}} \mathrm{d} s + (\lambda \eta^2 b - 2 r \lambda \eta S_0) (T - t )  \\ 
& = - 2 \lambda \eta^2 k \ln \frac{\xi   - 1}{\xi  e^{\varphi(T-t)} - e^{-\varphi(T-t)}} + (\lambda \eta^2 b - 2 r \lambda \eta S_0) (T - t ) \, .
\end{split}
\end{equation*}
Therefore, 
\begin{equation*}
h_0(t) = 2 \lambda \eta^2 k \ln \frac{\xi - 1}{\xi  e^{\varphi(T-t)} - e^{-\varphi(T-t)}}  + (\lambda \eta^2 b - 2 r \lambda \eta S_0) (t-T) \, , 
\end{equation*}
hence the results. 
\end{proof}

\begin{theorem}[verification theorem: finite-time-horizon control problem] \label{verification finite}
Let $u$ be given by the equation~\eqref{solution} in Proposition~\ref{solution to the finite model} . Then $u$ is the value function to the control problem~\eqref{finite time model}. Moreover, the optimum is achieved by the admissible optimal Markov control $\nu^{*}$ given by the feedback form 
\begin{equation} \label{feedback control finite time model}
\nu^*(t, q) =  \varphi \frac{\xi e^{2 \varphi (T-t)} + 1}{ \xi e^{2 \varphi (T-t)} - 1} q \,.
\end{equation}
\end{theorem}

\begin{proof}
Let $h_0, h_2 \in C^1([0, T], \mathbb R)$ be given by the equation~\eqref{solution h_2 h_0 finite}, and $u(t,q) = h_0(t) + h_2(t) q^2$ be the candidate value function to the control problem in the finite-time setting. As the HJB equation~\eqref{finite hjb} attains its supremum at $\nu^*$ given by~\eqref{sup of control}, hence 
\begin{equation*}
\nu^* =  - \frac{b + 2 h_2(t)}{2k} q =  \varphi \frac{\xi  e^{2 \varphi (T-t)} + 1}{\xi   e^{2 \varphi (T-t)} - 1} q \, .
\end{equation*}
Let $\varrho(t) = \varphi \frac{\xi  e^{2 \varphi (T-t)} + 1}{\xi   e^{2 \varphi (T-t)} - 1}$. Clearly, the control $\nu^*$ is a measurable function of $t$ and $q$, and $\nu^*$ is $\mathbb F-$adapted. 
To prove $\nu^* \in \mathcal{A}[t,T]$, we need to show that $\nu^*$ is square integrable. 
Due to the fact that 
\begin{equation*}
    \mathbb E[\int_t^T |\nu^*|^2 \mathrm{d} s] \leq \mathbb E[\int_t^T |\sup_{s \in [0,T]} \varrho(s)|^2 |Q_s^{\nu^*}|^2 \mathrm{d} s] \leq L^2 \mathbb E[\int_t^T |Q_s^{\nu^*}|^2 \mathrm{d} s] \, ,
\end{equation*}
where $L = \varphi + \frac{2 \varphi}{\xi -1 } > 0$ is a constant independent of time, it is enough to show that $Q_t^{\nu^*}$ is square integrable. 
For notational simplicity we omit the superscript $\nu^*$ of $Q^{\nu^*}_t$. 
Since the dynamics of $Q_t$ is 
\begin{equation*}
\begin{split}
d Q_t & = - \nu^*(Q_t) \mathrm{d} t + \zeta^+ \mathrm{d} N_t^+ - \zeta^- \mathrm{d} N_t^- \\ 
& = \big( - \nu^*(Q_t)  + \lambda^+ \zeta^+ - \lambda^- \zeta^- \big) \mathrm{d} t + \zeta^+ \mathrm{d} \tilde N_t^+ - \zeta^- \mathrm{d} \tilde N_t^- \, . 
\end{split}
\end{equation*}
As $\mathbb E[\lambda^\pm \zeta^\pm] = \lambda \eta$ and $\mathbb E[\zeta^{\pm2}] < + \infty$ and by using the Gr\"{o}nwall's inequality, for any $s \in [t,T]$, we have
\begin{equation*}
\begin{split}
\mathbb E \left| Q_s \right|^2 & \leq 4 q^2 + 4 (T-t) \mathbb E \big[\int_t^s \left|\nu^*(Q_u) \right|^2 \mathrm{d} u \big]  \\ 
& \leq 4q^2 \exp \big( (T- t)\int_t^s \varrho (u) \mathrm{d} u \big) \leq 4q^2 \exp \big( L (T- t)^2\big) \, , 
\end{split}
\end{equation*}
where $L = \varphi + \frac{2 \varphi}{\xi -1 }$ is independent of time. Therefore, we have 
\begin{equation*}
\mathbb E \big[\int_t^T Q_s^2 \, \mathrm{d} s \big] \leq \int_t^T \sup_{s \in [t,T]}\mathbb E[ |Q_s|^2] \, \mathrm{d} s \leq (T-t) 4q^2 \exp \big( L (T- t)^2\big)   < +\infty \, .
\end{equation*}

Now we need to show that $u(t,q) = v(t,q)$, i.e. $u$ is the value function of the control problem~\eqref{finite time model} and the control $\nu^*$ is indeed the optimal control. To that purpose, let us consider a control $\nu \in \mathcal{A}[t, T]$.
From the It\^o's lemma, we have
\small
\begin{equation*}
\begin{split}
u(T, Q_T) & = u(t, q)  + \int_t^T  \partial_s u(s, Q_s) \, \mathrm{d} s - \int_t^T   \nu \partial_q u(s, Q_{s}) \, \mathrm{d} s \\
& \quad + \int_t^T \big[ u(s, Q_{s^-}+\zeta^+) - u(s, Q_{s^-})\big] \, \mathrm{d} sN_s^+  + \int_t^T \big[ u(s, Q_{s^-}-\zeta^-) - u(s, Q_{s^-}) \big] \, \mathrm{d} N_s^- \\
& = u(t, q)  + \int_t^T  \partial_s u(s, Q_{s}) \, \mathrm{d} s - \int_t^T   \nu \partial_q u(s, Q_{s}) \, \mathrm{d} s \\
& \quad + \int_t^T \lambda \big[ u(s, Q_{s^-}+\zeta^+) - u(s, Q_{s^-})\big] \, \mathrm{d} s + \int_t^T \lambda \big[ u(s, Q_{s^-}-\zeta^-) - u(s, Q_{s^-}) \big] \, \mathrm{d} s \\
& \quad + \int_t^T \big[ u(s, Q_{s^-}+\zeta^+) - u(s, Q_{s^-})\big] \, \mathrm{d} \tilde N_s^+  + \int_t^T \big[ u(s, Q_{s^-}-\zeta^-) - u(s, Q_{s^-}) \big] \, \mathrm{d} \tilde N_s^- \, , 
\end{split}
\end{equation*}
\normalsize
where the process $Q_t$ starts at time $t$ with initial values $q$ and $\tilde N_s^\pm$ are independent compensated Poisson processes. 
Now we have to ensure that the last two integrals consist of martingales so that their expectations are $0$. This can be concluded by using~\cite[Theorem 11.4.5]{shreve2004stochastic} and the fact that $Q_{s^-}$ is a left-continuous stochastic process and square integrable under the control $\nu \in \mathcal{A}[t,T]$. Take the expectation of $u(T, Q_T)$, hence, we have 
\begin{equation*}
\begin{split}
\mathbb E \big[u(T, Q_T) \big] & = \mathbb E \big[ G(Q_T)\big]= u(t,q) +  \mathbb E \Big[  \int_t^T  \big\{ \partial_s u(s, Q_{s})  -   \nu \partial_q u(s, Q_{s})  \\ 
& \quad + \lambda \big[ u(s, Q_{s^-}+\zeta^+) - u(s, Q_{s^-})\big]  +  \lambda \big[ u(s, Q_{s^-}-\zeta^-) - u(s, Q_{s^-}) \big] \big\} \, \mathrm{d} s \Big] \, .
\end{split}
\end{equation*}

Given that $\nu^*$ attains the equality in~\eqref{finite hjb}, hence, for any time $s \in [t,T]$ the following inequality holds
\begin{equation*}
\begin{split}
- F(Q_s, \nu) \geq & \partial_t u(s, Q_s) - \nu \partial_q u(s, Q_s) 
\\ & + \lambda \mathbb E\big[u(s, Q_{s^-} + \zeta^+) - u(s,Q_{s^-}) \big]+ \lambda \mathbb E\big[u(s, Q_{s^-} - \zeta^-) - u(s,Q_{s^-}) \big] \, .
\end{split}
\end{equation*}
Therefore, we have 
\begin{equation*}
\begin{split}
\mathbb E \Big[ \int_t^T F(Q_s^{\nu}, \nu) \,  \mathrm d s +  G(Q_T^{\nu}) \Big| Q_t = q\Big] \leq u(t,q) \, ,
\end{split}
\end{equation*}
where the equality is obtained with $\nu^*$ which is an admissible control. 
\end{proof}

\subsection{Discounted infinite-time-horizon control problem~\eqref{discounted infinite model}}
\label{proof theorem: existence of discounted infinite}
\begin{theorem} \label{existence of discounted infinite}
Let $h_0, h_2 \in \mathbb{R}$ be given by 
\begin{equation}
\begin{split}
h_2 & = \frac{k \beta - b}{2} - \frac{\sqrt{(k \beta - b)^2 + (4 k \phi - b^2)}}{2} \, , \\
h_0 & = \frac{2 \lambda \eta }{\beta} (\eta h_2 + r S_0) \, ,
\end{split}
\end{equation}
where we choose the discount factor $\beta$ such that $(k \beta - b)^2 + (4 k \phi - b^2) \geq 0$. 
Then the function $u : \mathbb R \to \mathbb R$ given by 
\begin{equation} \label{solution2}
u(q) = h_0 + h_2 q^2
\end{equation}
solves the HJB equation associated with the control problem~\eqref{discounted infinite model} which is given by
\begin{equation} \label{discounted infinite hjb}
0 =  - \beta u(q) + \sup_{\nu \in \mathcal{A}} \Big\{ - \nu \partial_q u + F(\nu, q) + \lambda \mathbb E\big[u( q + \zeta^+) - u(q) \big]+ \lambda \mathbb E\big[u( q - \zeta^-) - u(q) \big]  \Big\} \, ,
\end{equation}

Moreover, $u$ given by the equation~\eqref{solution2} is the value function to the control problem~\eqref{discounted infinite model}. 
\end{theorem}

\begin{proof}
By dynamic programming principle, the associated HJB equation for the control problem~\eqref{discounted infinite model} is 
\begin{equation*}
0 =  - \beta u(q) + \sup_{\nu} \Big\{ - \nu \partial_q u + F(\nu, q) + \lambda \mathbb E\big[u( q + \zeta^+) - u(q) \big]+ \lambda \mathbb E\big[u( q - \zeta^-) - u(q) \big]  \Big\} \, . 
\end{equation*}

Again by using the ansatz $u(q) = h_0 + h_1 q + h_2 q^2$, we obtain the following coupled system
\begin{equation*}
\begin{split}
- \beta h_0 + \frac{1}{4k} h_1^2 + 2 \lambda \eta r S_0  + 2 \lambda \eta^2 h_2 & = 0\, , \\
- \beta h_1 + \frac{1}{2k}(2 h_2 + b) h_1 & = 0 \,, \\
- \beta h_2 + \frac{1}{4k} (2h_2 + b)^2 - \phi & =0   \, .
\end{split}
\end{equation*}

We start with solving the last equation. Since $\beta$ is chosen to satisfy $(k \beta - b)^2 + (4 k \phi - b^2) \geq 0$, there are $2$ real roots of the quadratic equation, i.e. $h_2 = \frac{k \beta - b}{2}  \pm \frac{\sqrt{(k\beta - b)^2 + (4 k \phi  - b^2)}}{2}$. Note that the optimal feedback control is given by the form~\eqref{sup of control}. To guarantee that the optimal control is admissible, we keep $h_2 = \frac{k \beta - b}{2} - \frac{\sqrt{(k \beta - b)^2 + (4 k \phi - b^2)}}{2}$ since, in this case,  
\begin{equation} \label{infinite time optimal control}
\nu^{*}(q) = \frac{-\big(b q + \partial_q u(q) \big)}{2k} = \frac{\sqrt{(k \beta - b)^2 + (4 k \phi - b^2)} - k \beta}{2 k} q \, ,
\end{equation}
by choosing $2 \phi > \beta b$, we have $\varrho := \frac{\sqrt{(k \beta - b)^2 + (4 k \phi - b^2)} - k \beta}{2 k} > 0$. Hence the process $Q_t^{\nu^*}$ given by the dynamics~\eqref{eq:open_positions} under the control $\nu^*$ is square integrable since
$$ \sup_{t \in [0, \infty)} \mathbb E |Q_t^{\nu^*}|^2 \leq q^2 \vee \frac{\lambda (\mathbb E[|\zeta^{+}_1|^2] + \mathbb E[|\zeta^{-}_1|^2]}{2 \varrho} < + \infty \, . $$
Therefore, the Markov optimal control $\nu^*$ given by~\eqref{infinite time optimal control} is square integrable. Clearly, $\nu^*$ is $\mathbb F-$progressively measurable, hence $\nu^* \in \mathcal{A}$. 

We then can uniquely solve that $h_1 = 0$ and $h_0 = \frac{2 \lambda \eta r S_0  + 2 \lambda \eta^2 h_2}{\beta}$. 
Furthermore, by using a standard verification argument, which proceeds analogously to the proof of Theorem~\ref{verification finite} in Section~\ref{proof of prop: solution to the finite model}, we conclude the result. 
\end{proof}

\subsection{Proof of Theorem~\ref{ergodic constant from discounted infinite}} \label{proof theorem: ergodic constant from discounted infinite}

\begin{proof}
We start with $\lim_{\beta \rightarrow 0} \beta v_{\beta}(q)$. We know that $(k \beta - b)^2 + (4 k \phi - b^2) \geq 0$ when $\beta \rightarrow 0$ since $k , \phi \geq 0$. Therefore, from Theorem~\ref{existence of discounted infinite}, we have
\begin{equation*}
\begin{split}
\lim_{\beta \rightarrow 0} \beta v_{\beta}(q) & = \lim_{\beta \rightarrow 0} \beta (h_0 + h_2 q^2) \\
& = \lim_{\beta \rightarrow 0} 2 \lambda \eta r S_0  + 2 \lambda \eta^2 h_2 \\
& = 2 \lambda \eta r S_0  - \lambda \eta^2 b - 2 \lambda \eta^2 \sqrt{k \phi} \, .
\end{split}
\end{equation*}
Let $\hat \gamma = 2 \lambda \eta r S_0  - \lambda \eta^2 b - 2 \lambda \eta^2 \sqrt{k \phi} \in \mathbb R$, hence $\hat \gamma = \lim_{\beta \rightarrow 0} \beta v_{\beta}(q)$. 

On the other hand, from Theorem~\ref{solution to the finite model} for the control problem~\eqref{finite time model}, we have 
\begin{equation*}
v(0,q;T) = 2 \lambda \eta^2 k \ln \frac{\xi   - 1}{\xi  e^{\varphi T} - e^{-\varphi T}} - (\lambda \eta^2 b - 2 r \lambda \eta S_0) T  + (k \varphi  \frac{1 + \xi   e^{2 \varphi T}}{1 - \xi   e^{2 \varphi T}} - \frac{1}{2}b )q^2 \, ,
\end{equation*}
\normalsize
with $\varphi = \sqrt{\frac{\phi}{k}}$ and $\xi   = \frac{\alpha - \frac{1}{2}b + \sqrt{k\phi}}{\alpha - \frac{1}{2}b - \sqrt{k\phi}}$. Therefore, 
\begin{equation*} 
\begin{split}
\lim_{T \rightarrow + \infty}  \frac{1}{T} v(0, q; T) 
& = - 2 \lambda \eta^2 k  \lim_{T \rightarrow + \infty}  \frac{1}{T} \ln(\xi   e^{\varphi T}) + (2 r \lambda \eta S_0 - \lambda \eta^2 b) \\
& = 2 \lambda \eta r S_0 - \lambda \eta^2 b  - 2 \lambda \eta^2 \sqrt{k\phi} \\ 
& = \hat \gamma \, .
\end{split}
\end{equation*}

Finally we need to show that the ergodic constant $\gamma$ given by~\eqref{reduced ergodic control problem} is equal to $\hat \gamma$, which the proof is analogous to that of~\cite[Theorem 5]{cao2024logarithmic} and thus omitted here for brevity.
\end{proof}

\subsection{Proof of Corollary~\ref{ergodic control}} \label{proof prop optimal ergodic control}
\begin{proof}
From Theorem~\ref{existence of discounted infinite}, it is obvious that, for any initial state $q \in \mathbb R$, $\left|\beta v_\beta(q)\right|$ and $\left| v_\beta(q) - v_\beta(0) \right|$ are bounded. Hence, by following the proof in~\cite[Appendix A.3]{cao2024logarithmic}, the ergodic HJB equation for the control problem~\eqref{reduced ergodic control problem} is 
\begin{equation*}
0 =  - \gamma + \sup_{\nu} \Big\{ - \nu \partial_q u + F(\nu, q) + \lambda \mathbb E\big[u( q + \zeta^+) - u(q) \big]+ \lambda \mathbb E\big[u( q - \zeta^-) - u(q) \big]  \Big\} \, ,
\end{equation*}
where $\gamma$ is the ergodic constant given by Theorem~\ref{ergodic constant from discounted infinite} and $u$ is the relative value function uniquely defined up to a constant. 
Again the supremum is attained at $\nu^*$ given by the formula~\eqref{sup of control} and by using the ansatz $u(q) = h_0 + h_1 q + h_2 q^2$, we have the following coupled system
\begin{equation*}
\begin{split}
\lambda \eta^2 b + 2 \lambda \eta^2 \sqrt{ k \phi} + \frac{h_1^2}{k} + 2 \lambda h_2 \eta^2 & = 0 \,, \\
\frac{1}{2k} (2h_2 + b ) h _1 q & = 0 \, ,\\
\big( \frac{1}{4k} (2h_2 + b)^2 - \phi \big)q^2 & = 0 \, .
\end{split}
\end{equation*}
To guarantee that the optimal control is in the set of admissible controls, we keep the root $h_2 = - \sqrt{k \phi} - \frac{1}{2}b$. Hence the optimal ergodic control $\nu^*$ given in the feedback form is 
\begin{equation*}
\nu^* (q) = \frac{-(b(q) + \partial_q u(q))}{2k} = \sqrt{\frac{\phi}{k}} q \, . 
\end{equation*}
To prove $\nu^* \in \mathcal{A}$, it proceeds analogously to the proof in Section~\ref{proof theorem: existence of discounted infinite}. We then
solve $h_1 = 0, h_0 \in \mathbb R$ to the coupled system. Therefore, we can see that the relative ergodic value function $u(q) = h_0 + h_2 q^2$ is unique up to a constant. 
\end{proof}

\raggedright
\bibliographystyle{abbrv}
\bibliography{Bibliography}

\end{document}